\documentclass{article}

\usepackage[utf8]{inputenc}
\usepackage{fullpage}
\usepackage{pgfplots}
    \pgfplotsset{compat=1.14}
\usepackage{tikz}
    \usetikzlibrary{shapes,decorations}
    \usetikzlibrary{positioning}
\usepackage{graphicx}
\usepackage{mathrsfs, amsthm, mathtools, amssymb}
\usepackage{algorithm2e}
\usepackage{enumitem}
\usepackage{authblk}
\usepackage{natbib}
\usepackage{nicefrac}
\usepackage{hyperref}
    
\usepackage[capitalize,noabbrev]{cleveref} % TC: KEEP IT LAST!

\newcommand{\cD}{\mathcal{D}}

\newcommand{\cO}{\mathcal{O}}

\newcommand{\cQ}{\mathcal{Q}}

\newcommand{\cU}{\mathcal{U}}

\newcommand{\E}{\mathbb{E}}

\newcommand{\I}{\mathbb{I}}

\newcommand{\N}{\mathbb{N}}

\newcommand{\R}{\mathbb{R}}

\newcommand{\e}{\varepsilon}

\newcommand{\lrb}[1]{\left(#1\right)}
\newcommand{\brb}[1]{\bigl(#1\bigr)}
\newcommand{\Brb}[1]{\Bigl(#1\Bigr)}

\newcommand{\lsb}[1]{\left[#1\right]}
\newcommand{\bsb}[1]{\bigl[#1\bigr]}
\newcommand{\Bsb}[1]{\Bigl[#1\Bigr]}

\newcommand{\lcb}[1]{\left\{#1\right\}}
\newcommand{\bcb}[1]{\bigl\{#1\bigr\}}

\newcommand{\lce}[1]{\left\lceil#1\right\rceil}

\newcommand{\labs}[1]{\left\lvert#1\right\rvert}

\newcommand{\Babs}[1]{\Bigl\lvert#1\Bigr\rvert}

\newcommand{\bno}[1]{\bigl\lVert#1\bigr\rVert}

\DeclareMathOperator*{\argmax}{argmax}
\newcommand{\dif}{\,\mathrm{d}}

\newcommand{\lip}{Lipschitz}
\newcommand{\leb}{Lebesgue}
\newcommand{\leib}{Leibniz}
\newcommand{\BR}{\mathrm{BR}}
\newcommand{\fracc}[2]{#1/#2}

\newcommand{\iop}{\infty}
\newcommand{\mypapertitle}{Online Learning in Supply-Chain Games}

\newcommand{\stack}{Stackelberg}

\newcommand{\us}{u_{\mathrm{S}}}
\newcommand{\ur}{u_{\mathrm{R}}}

\newcommand{\piya}{Piyavskii--Shubert}
\newcommand{\expthree}{Exp3-VI}

\newcommand{\loss}{\ell}
\newcommand{\hloss}{\widehat \ell}
\newcommand{\dd}{\dif}

\newtheorem{assumption}{Assumption}

\newtheorem{theorem}{Theorem}
\newtheorem{lemma}{Lemma}
\newtheorem{proposition}{Proposition}

\begin{document}

\title{\mypapertitle
\texorpdfstring{\thanks{
Partially supported by the Italian MIUR PRIN 2017 Project ALGADIMAR Algorithms, Games, and Digital Markets.
Marco Scarsini is a member of INdAM-GNAMPA.}}{}
}

\author[1]{Nicol\`o Cesa-Bianchi}
\author[1,2]{Tommaso Cesari}
\author[3]{Takayuki Osogami}
\author[4]{Marco Scarsini}
\author[5]{Segev Wasserkrug}

\affil[1]{Universit\`a degli Studi di Milano, Milan, Italy}
\affil[2]{Toulouse School of Economics, Toulouse, France}
\affil[3]{IBM Research - Tokyo, Tokyo, Japan}
\affil[4]{Luiss University, Rome, Italy}
\affil[5]{IBM Research - Haifa, Haifa, Israel}

\maketitle

\begin{abstract}
We study a repeated game between a supplier and a retailer who want to maximize their respective profits without full knowledge of the problem parameters. After characterizing the uniqueness of the Stackelberg equilibrium of the stage game with complete information, we show that even with partial knowledge of the joint distribution of demand and production costs, natural learning dynamics guarantee convergence of the joint strategy profile of supplier and retailer to the Stackelberg equilibrium of the stage game. We also prove finite-time bounds on the supplier's regret and asymptotic bounds on the retailer's regret, where the specific rates depend on the type of knowledge preliminarily available to the players. In the special case when the supplier is not strategic (vertical integration), we prove optimal finite-time regret bounds on the retailer's regret (or, equivalently, the social welfare) when costs and demand are adversarially generated and the demand is censored.
\end{abstract}

\section{Introduction}
\label{se:intro}

The newsvendor problem is a central topic in inventory theory and, more generally, in the analysis of supply chains. In its classical version 
\cite{ArrHarMarE1951},
a retailer orders a certain quantity of a perishable good from a supplier.
The decision of how much to order is made before the realization of the unknown demand for the good.  
If all costs are linear, the optimal decision is a quantile of the demand distribution that depends on the parameters of the model (i.e., the wholesale price charged by the supplier and the retail market price).
Even in this simple framework, it is clear that the retailer can compute the optimal quantity only if the demand distribution and the model parameters are known.

In a practical newsvendor scenario, it is rarely the case that the retailer is the only decision maker. 
For instance, the wholesale price could be determined by a supplier incurring exogenous production costs that are unknown to the retailer.
These multi-agent versions of the newsvendor problem can be analyzed via a game-theoretic approach, where
the optimal choices are expressed in terms of equilibria of a game.
As it is unreasonable to assume that supplier and retailer have full knowledge of the distributions of production costs and demand,
it is important to find strategies that perform well even when only partial knowledge of these quantities is available to each player.

\paragraph{Our contribution.}
\label{suse:our-contribution}
In this paper we consider a competitive newsvendor model where players do not have full knowledge of the relevant parameters.
The problem is modeled as a repeated game between a supplier and a retailer, whose goal is to minimize their regret.

We start by considering the stage game with complete information, which we model as a Stackelberg game. Here, the supplier chooses a wholesale price and reveals it to the retailer, who in turn chooses a quantity to order by solving a newsvendor problem.
Under weak conditions (see \Cref{ass:distribution-stage}) on the joint distribution of production costs $C$, retail price $P$, and demand $D$, we characterize (\Cref{t:unique-stack}) the uniqueness of the Stackelberg equilibrium (SE) of the game, which we show to be in pure strategies.

We then consider a repeated game in which the supplier only knows the marginal distribution of $C$, and the retailer only knows the marginal of $(P,D)$. Assuming that, at each time $t$, the retailer chooses the quantity $q_t$ by best responding to the supplier's choice of wholesale price $w_t$, in \Cref{t:strong-assumptions}, we show that the supplier's average regret vanishes at rate $T^{-1/2}$, while the players' strategy profile $(w_t,q_t)$ converges to the SE of the stage game asymptotically at the same $T^{-1/2}$ rate.
If the supplier is given an upper bound on the Lipschitz constant of their own utility, then the supplier's average regret vanishes at a faster rate $(\ln T)/T$. 
Notably, in this case the rate of convergence to the SE of the players' strategy profile can be arbitrarily slow (\Cref{t:strong-assumptions+lip}). 
Finally, if in the previous setting we drop the assumption that the retailer knows $(P,D)$, then the supplier's regret vanishes at rate $T^{-1/3}$ provided the retailer best responds to $w_t$ based on the empirical distribution of past realizations $(P_s,D_s)$ for all $s < t$ (\Cref{th:etc-ftl}).

The last part of this work deals with vertical integration, a setting in which the retailer is the only decision-maker because the supplier always sells at production cost. This can happen for a number of reasons; for instance, because the supply chain is owned by the same company, or because both parties preliminarily agreed on a revenue-sharing contract. In \Cref{th:buyer} we prove finite-time regret bounds for adversarial (as opposed to stochastic) sequences of production costs and demand functions. Our result applies to a setting in which the retailer is simultaneously optimizing the order quantity and the retail price in the presence of censored feedback. Our bounds show that the average regret vanishes at rate $T^{-1/3}$ (ignoring logarithmic factors). As this setting includes posted-price auctions as a special case, this rate is not significantly improvable.

Our analyses of convergence of the learning dynamics to the SE borrow ideas from online learning and bandits (e.g., Follow-The-Leader and Explore-Then-Commit), and use zeroth-order optimization (e.g., the Piyavskii–Shubert algorithm). The adversarial analysis in the vertical integration setting builds on the techniques developed in \cite{cesa2017algorithmic}, which in turn are an adaptation of the general framework introduced by Alon et al.~\cite{alon2017nonstochastic}.

\paragraph{Related work.}
\label{suse:related-literature}
The newsvendor problem, also known as the newsboy problem, goes back to Edgeworth \cite{Edg:JRSS1888}.
The formalization used in this work is due to Arrow et al.~\cite{ArrHarMarE1951}; we refer the reader to the handbook \cite{Cho:HNP2012} for a survey of the many variants of Arrow's model.

A game-theoretic formulation of the newsvendor problem with competing retailers is proposed by Parlar~\cite{Par:NRL1988}---see also \cite{LipMcC:OR1997,Mahvan:OR2001,NetRudWan:IIET2006}.
Wang and Gerchak \cite{WanGer:MSOM2003} use a Stackelberg game to model a situation where an assembler has to buy components from different suppliers.
Lariviere and Porteus \cite{LarPor:MSOM2001} study a model where a supplier and a retailer interact through a price-only contract, and compare its efficiency with the efficiency of an integrated system. 
Adida and DeMiguel~\cite{AdiDeM:OR2011} consider a competitive inventory model with several suppliers and several retailers, and prove equilibrium uniqueness under some symmetry conditions.
We refer the reader to Cachon and Netessine~\cite{CacNet:INFORMS2006} for a survey of the literature on game-theoretic models in supply chain analysis.

The problem of learning equilibria is investigated by Balcan et al.~\cite{balcan2015commitment} for Stackelberg security games. They prove bounds on the leader's regret when the follower has a type that changes over time in a known and finite class. Their results hold in both the full information setting (where the leader can observe the type of the follower) and in the bandit setting (where the follower's type is not observed).
Sessa et al.~\cite{sessa2020learning} study general repeated games between a leader with a finite number of actions and an follower with a finite number of types which may adversarialy change over time. They prove bounds on the leader's regret in the full information setting when the utility of each follower (which is determined by its type) is only known to satisfy certain regularity assumptions.
Bai et al.~\cite{bai2021sample} show that in the bandit setting with finitely many actions for leader and follower, there exist expected utility functions such that any leader's algorithm suffers non-vanishing regret with probability at least $1/3$. They also show leader algorithms that converge to SE up to a certain suboptimality gap.
Deng et al.~\cite{deng2019strategizing} prove some interesting non-constructive results.
Let $V$ be the utility of the leader in a SE. Under some mild assumptions, they show that for any $\e > 0$ the leader can always obtain a utility of at least $(V - \e)T - o(T)$ in $T$ rounds, against any no-regret algorithm of the follower (note that the convergence rate is not explicit in their results).
Mansour et al.~\cite{mansour2022strategizing} extend these results to Bayesian games.

Note that our results take advantage of the specific structure of the utility functions to obtain good rates for the leader's regret in a bandit setting. Note also that, unlike previous works, the follower's best response in our setting is not determined by a type, but rather learned from observed data. This allows us to prove that the follower's regret vanishes too. As a consequence, we are also able to prove convergence to SE of the players' strategy profile.

Often, quantities and prices between suppliers and retailers are determined by contracts.
Cachon and Lariviere \cite{CacLar:MS2005} study revenue-sharing contracts in general supply chain models, and show that such contracts result in the same overall efficiency as vertical integration. Note that under vertical integration the retailer is the only decision-maker, and our setting becomes similar to a single-agent newsvendor model.
Among the many works investigating learning approaches to the newsvendor model, the ones most relevant to our work are the regret minimization analyses of Huh and Rusmevichientong~\cite{huh2009nonparametric} and of Besbes and Muharremoglu~\cite{besbes2013implications}. However, to the best of our knowledge, the adversarial newsvendor problem with censored demand is only studied by Lugosi et al.~\cite{lugosi2021hardness}. The reason why their regret rates are better than ours is a different choice of the retailer's decision space. In their setting, the retailer only chooses the order quantity (while we also have control on the retail price, which affects the demand). Moreover, their decision space is finite while ours is continuous.

\section{Stage game and unique \stack{} equilibrium}
\label{s:stage-full-section}

In this section, we analyze a (one-shot) stage game of supply chain and characterize the uniqueness of its SE (formally  defined below). We also provide some insights on the learning results proven in the following sections.

\paragraph{The stage game.}
\label{s:stage}
An instance of the stage game is characterized by a known distribution $\cD$ on $[0,\iop)^3$ that governs the (possibly correlated) \emph{production cost} $C$ of the supplier, the \emph{retail price} $P$ dictated by the market, and the \emph{demand} $D$. 
We make the following assumption on $\cD$.
\begin{assumption}
\label{ass:distribution-stage}
The distribution $\cD$ of $(C,P,D) \in [0,\infty)^3$ satisfies the following:
\begin{enumerate}
    \item \label{i:well-def} $\E[C]$, $\E[P]$, $\E[D]$, and $\E[PD]$ are all finite.
    \item \label{i:econ} $\E[C] < \E[P]$.
    \item \label{i:tech} The conditional distribution of $D$ given $(C,P)$ admits a density (w.r.t.\ the \leb{} measure) such that $f( d \mid c,p ) > 0$, for all $(c,p,d)\in [0,\iop)^3$.
\end{enumerate}
\end{assumption}
\Cref{i:well-def} guarantees that the expected utilities of the supplier and the retailer (see below for a definition) are well-defined and finite for any action profile. \Cref{i:econ} is an economic assumption stating that, on average, the supplier's cost is lower than the retail price, thus eliminating trivial scenarios. \Cref{i:tech} is a mild technical condition that simplifies the presentation of the proof of \Cref{t:unique-stack}.

We denote the conditional cumulative distribution function and survival function of the demand, given the supplier's cost and retail price, by
\[
    F(d\mid c,p) \coloneqq \int_0^d f( x \mid c,p ) \dif x
    \quad \text{and} \quad
    \bar F (d\mid c,p) \coloneqq 1 - F (d\mid c,p)
    \quad
    \forall c,p,d\ge 0.
\]
In this section, we assume that the structure of the model (namely, $\cD$ in \Cref{ass:distribution-stage}) is common knowledge to both players.
The game proceeds as follows.
First, the supplier (S) selects a wholesale price $w \in [0,\iop)$ and reveals it to the retailer. 
Then, the retailer (R) selects a quantity $q\in[0,\iop)$.
Their expected \emph{utilities} are respectively defined, for any $(w,q) \in [0,\iop)^2$, by
\[
    \us(w,q) \coloneqq qw - q \E\bsb{ C }
\quad\text{and}\quad
    \ur(w,q) \coloneqq \E\bsb{ \min\{q,D\} P } - qw. 
\]
where the expectations are with respect to $(C,P,D)\sim\cD$.

Finally, the \emph{\stack{} equilibria} of this game are defined as strategy pairs $(w^\star, q^\star) \in [0,\iop)^2$ such that
\[
    w^\star \in \argmax_{w \in [0,\iop) } \us \brb{ w, \BR(w) }
\quad \text{and} \quad
    q^\star = \BR(w^\star).
\]
where $\BR$ is a best-response of the retailer, i.e., $\BR(w) \in \argmax_{q \in [0,\iop) } \ur(w,q)$.

\paragraph{Uniqueness of the \stack{} equilibrium.}
We can now state our characterization of the uniqueness of the SE in the stage game\footnote{For example, the SE is unique when $D$ has a Weibull distribution with nondecreasing failure rate and $(C,P)$ is deterministic (see \Cref{p:unique} in the appendix).}.

\begin{theorem}
\label{t:unique-stack}
Under \Cref{ass:distribution-stage}, let $g(w)\coloneqq h^{-1}(w)$ be the inverse of 
    \begin{align}
    \label{eq:hx}
        h(x) 
    \coloneqq
        \E \bsb{ P \bar  F(x \mid C,P) }.
    \end{align}
 Then the following conditions are equivalent:
\begin{enumerate}
    \item The stage game in \Cref{s:stage} admits a unique SE $\brb{ w^\star, \, g(w^\star) }$.
    \item \label{i:condition-unique-stack} $\{w^{\star}\} \equiv \argmax_{w\in A} g(w)\brb{w-\E[C]}$, where
    \[
        A 
    \coloneqq 
        \lcb{ w \in \brb{\E[C],\E[P]} : -\frac{g'(w)}{g(w)}
    =
        \frac{1}{w-\E[C]} }.
    \]
\end{enumerate}
\end{theorem}

To prove the theorem, we begin with a simple but key lemma, whose proof we defer to \Cref{s:proof-lemma-g}.
\begin{lemma}
    \label{l:proprerties-of-g-2}
    Under \Cref{ass:distribution-stage}, the function $h$ in \Cref{eq:hx}
    is differentiable on $x\in(0,\iop)$ and has a strictly negative derivative. and therefore
    is invertible, and $h^{-1}(w)$ is differentiable and strictly decreasing on $w\in\brb{ 0, \, \E[P] }$.
\end{lemma}

We are now ready to prove \Cref{t:unique-stack}.
\begin{proof}[Proof of \Cref{t:unique-stack}]
We prove the theorem by backward induction.
First, we show that under \Cref{ass:distribution-stage}, the retailer has a unique best response.
Then, we show that given that the retailer best responds, the supplier has a unique optimal move $w^\star$ if and only if the condition of the theorem holds.

\noindent\textbf{Retailer's move.} 
Fix an arbitrary wholesale price $w\ge 0$. The retailer's utility is, for any $q\ge 0$,
\begin{align*}
    \ur(w,q) 
&=
    \E \lsb{ P \, \E\lsb{ \int_0^q x f(x \mid C,P) \dif x + q \brb{ 1 - F(q \mid C,P) } \mid C,P } } - qw.
\end{align*}
To maximize it, we compute its derivative, which is justified by \Cref{ass:distribution-stage} combined with the \leib{} integral rule. 
For any $q>0$, we obtain
\begin{equation}
\label{e:derivative-of-retailer-objective}
    \frac{\partial}{\partial q} \ur(w,q) 
=
    \E \Bsb{ P \, \E \lsb{ \bar F(q \mid C,P) \mid C,P} } - w,
\end{equation}
which is non-negative if and only if
$
    \E\bsb{ P \bar{F}(q \mid C,P) } \ge w
$.
By the arbitrariness of $w$ and \Cref{l:proprerties-of-g-2}, we can conclude that, for any choice of the wholesale price $w > 0$, there exists a unique maximizer $q^\star_w = \BR(w)$ of $q \mapsto \ur(w,q)$ where
\begin{equation}
\label{e:best-response}
    \BR(w) \coloneqq
    \begin{cases}
    g(w) & \text{if $w < \E[P]$}, \\
    0 & \text{if $w \ge \E[P]$}. \\
    \end{cases}
\end{equation}
\textbf{Supplier's move.}
Given the retailer's best response $q^\star_w = \BR(w)$, the supplier's utility is, for
any $w > 0$,
\begin{equation}
\label{e:zeroth-order-supplier}
    \us(w,q^\star_w)
=
    q^\star_w\brb{ w-\E[C] }.
\end{equation}
Since $q_w^\star=0$ for all $w\ge\E[P]$ and $w-\E[C]\le 0$ for
all $w\le\E[C]$, to maximize the supplier's expected utility, we
can restrict our search to $w\in\brb{\E[C],\,\E[P]}$, where
$
    \us(w,q^\star_w)
=
    g(w) \brb{ w-\E[C] }
$
is strictly positive and differentiable.
To find the maximum, then, we can study the sign of the derivative of the supplier's
expected utility, obtaining, for all $w\in\brb{\E[C], \, \E[P]}$,
\begin{equation}
    \label{e:derivative-of-supllier-objective}
    \frac{\partial}{\partial w} \us(w,q_w^\star)
=
    g'(w) \brb{w-\E[C]} + g(w).
\end{equation}
Thus, the set $A$ in the statement of the theorem is the set of all stationary points of $\us(w,q_w^\star)$ and the condition that $\argmax_{w\in A} g(w)\brb{w-\E[C]}$ is a singleton is exactly stating that there exists a unique maximizer of $g(w)\brb{w-\E[C]} = \us(w,q_w^\star)$, which coincides with the existence of a unique SE.
\end{proof}
In general, the payoffs under SE may be unique under weaker conditions than the ones for the uniqueness of the SE, but such weaker conditions do not appear to be simply stated in our case.

We also note that efficiency of the Stackeleberg equilibrium can be measured using the price of anarchy. 
Some preliminary results can be found in \Cref{s:inefficiency}.

\section{Learning the Stackelberg equilibrium}
\label{s:online-learning}

A limitation of \Cref{t:unique-stack} is that, even when a stage game has a unique SE, in order to compute it, both players need to know the underlying distribution $\cD$.
In this section, we show how to circumvent this issue and achieve convergence to the unique SE without relying on the knowledge of $\cD$. We do this by reconstructing the salient features of $\cD$ through a learning technique in a repeated game.
An instance of the repeated game is characterized by a distribution $\cD$ on\footnote{In contrast to \Cref{s:stage-full-section}, we assume here that $\cD$ is bounded (without loss of generality, by $1$). This assumption is for simplifying the presentation; all the following results can be extended to the unbounded case simply by assuming subgaussianity.} $[0,1]^3$, that governs the (possibly correlated) production cost, retail price, and demand.

We study the following online protocol.
At each round $t=1,2,\dots$:
\begin{enumerate}
\item Nature draws $(C_t,P_t,D_t)$ i.i.d.\ according to $\cD$.
\item The supplier (S) selects a wholesale price $W_t\in[0,1]$ and reveals it to the retailer (R).
\item The production cost $C_t$ is revealed to S.
\item R buys a quantity $Q_t \in [0,1]$ paying $Q_t W_t$ to S.
\item The retail price $P_t$ and the demand $D_t$ are revealed to R.
\item The market buys a quantity $\min\{Q_t,D_t\}$, paying $\min\{Q_t,D_t\}P_t$ to R.
\end{enumerate}

The individual goals of the supplier and the retailer are to maximize, for any time horizon $T$, their individual long-term expected utilities 
\[
	\sum_{t=1}^T \E \bsb{ \sigma(W_t;Q_t,C_t)}
\quad\text{and}\quad
	\sum_{t=1}^T \E \bsb{ \rho(Q_t;W_t,P_t,D_t) \mid W_t},
\]
where we define
$
	\sigma(w;q,c)
\coloneqq
	qw - qc
$
and
$
	\rho(q;w,p,d)
\coloneqq
	\min\{q,d\} p - qw
$
for $(w,q,c) \in [0,1]^3$ and $(q,w,p,d) \in [0,1]^4$.
The conditional expected utility of the retailer is to be maximized with high probability with respect to $(W
_t)_{t\in\N}$.
The asymmetry in the objectives of the supplier (S) and retailer (R) is due to the fact that R is revealed $W_t$ before making a decision at time $t$, while S has to act before observing $Q_t$.

\paragraph{Convergence to SE with partial information.}
\label{s:partial-info}
In \Cref{s:stage-full-section}, we showed that the stage game admits a unique SE under the assumptions in \Cref{t:unique-stack}, and one can obtain the SE if $\cD$ is known by both the supplier (S) and the retailer (R).
We now show that a repeated interaction between S and R who act rationally and selfishly can lead to a convergence to the SE, even assuming that S and R have only partial information on $\cD$.
In \Cref{s:convergence-no-info}, we will show that the same result can be obtained when S and R have essentially no information on $\cD$ at a cost of a slower convergence rate.
We make the following assumption.
\begin{assumption}
\label{ass:distribution-learning-strong}
The distribution $\cD$ of $(C,P,D) \in [0,1]^3$ satisfies the following:
\begin{enumerate}
    \item \label{i:econ-learning-strong} $\E[C] < \E[P]$.
    \item \label{i:econ-2-learning-strong} $D$ and $C$ are conditionally independent, given $P$.
    \item \label{i:tech-learning-strong} The conditional distribution of $D$ given $P$ admits a density (with respect to the \leb{} measure) such that $f( d \mid p ) > L$, for some $L>0$ and all $(p,d)\in [0,1]^2$.
    \item Condition~\ref{i:condition-unique-stack} of \Cref{t:unique-stack} holds (i.e., the SE is unique).
\end{enumerate}
\end{assumption}
Condition~\ref{i:econ-learning-strong} states that, on average, the supplier's cost is lower than the retail price, eliminating trivial scenarios. Condition~\ref{i:econ-2-learning-strong} states that the demand $D$ may depend on the supplier's cost $C$ only via the retail price $P$.
Condition~\ref{i:tech} is a mild technical condition guaranteeing that the learning problem is at least \lip{} (see below).
Essentially \Cref{ass:distribution-learning-strong} implies that the conditions corresponding to \Cref{ass:distribution-stage} are satisfied for $\cD$ with bounded support and guarantees that \Cref{t:unique-stack} can be applied, so that the learning problem is tractable.

The partial knowledge of the supplier and the retailer is formalized as follows.
\begin{assumption}
\label{ass:limited-info}
The supplier has access to the marginal distribution of $C$, but not to that of $(P,D)$; the retailer has access to the marginal of $(P,D)$, but not to that of $C$.
\end{assumption}
Namely, the supplier has more information on the production cost, while the retailer has more information about the direct interaction with the market.

The retailer's strategy is to best-respond to any wholesale price. Note that $q^\star_{w} = \BR(w)$ can be computed exactly via \Cref{e:best-response} thanks to Assumptions~\ref{ass:distribution-learning-strong}~and~\ref{ass:limited-info}.

Under these assumptions, and given that the retailer best-responds, the supplier can compute their expected cost $\E[C]$ and is left with solving a zeroth-order \lip{} optimization problem (i.e., maximizing \Cref{e:zeroth-order-supplier}) without the knowledge of the \lip{} constant of its objective.
This can be done with a simple explore-then-commit algorithm (\Cref{a:etc}), which we assume to be the supplier's strategy.

\begin{algorithm}
\DontPrintSemicolon
\caption{Explore-Then-Commit\label{a:etc}}
\textbf{input:} Time horizon $T$\;
\For{$t=1,\dots,\lfloor T^{1/2} \rfloor$}
{
    Select the wholesale price $w_t \coloneqq \fracc{t}{\big(\lfloor T^{1/2} \rfloor + 1\big)}$\;
    Observe the quantity $q_t$\;
}
\For{$t=\lfloor T^{1/2} \rfloor+1,\dots,T$}
{
    Select a wholesale price $w_t=w_{s^\star}$, where ${\displaystyle s^\star \in \argmax_{s=1,\ldots,\lfloor T^{1/2} \rfloor} q_s \brb{ w_s - \E[C] } }$\;
    Observe the quantity $q_t$\;
}
\end{algorithm}

\begin{theorem}
\label{t:strong-assumptions}
Under Assumptions~\ref{ass:distribution-learning-strong}~and~\ref{ass:limited-info}, for any horizon $T$, if the supplier runs \Cref{a:etc} with input $T$ and the retailer best responds, then
\begin{equation}
\label{e:regr-supp-1l} 
        \E \bsb{ \sigma(w^\star;q^\star,C) } - \frac1T \sum_{t=1}^T \E \bsb{ \sigma(w_t;q_t,C_t)} 
    \le 
        \left( \frac{1-\E[C]}{\E[P]L} +2 \right) T^{-1/2},
\end{equation}
where $(w^\star,q^\star)$ is the unique SE of the stage game.
Also, for all sufficiently large $T$:
\begin{enumerate}
    \item \label{i:regr-ret-1} 
    $
        \E \bsb{ \rho(q^\star;w^\star,P,D) } - \frac1T \sum_{t=1}^T \E \bsb{ \rho(q_t;w_t,P_t,D_t) }
    \le 
        \brb{ L^{-1} + 2 }T^{-1/2}.
    $
    \item \label{i:last-iterate-conv-1} 
    $
        \bno{ (w^\star,q^\star) - (w_T,q_T) }_1 
    \le 
        \brb{(\E[P] L)^{-1}+1} T^{-1/2}.
    $
\end{enumerate}
\end{theorem}
Item~\ref{i:last-iterate-conv-1} shows \emph{last-iterate} convergence to the unique SE (in contrast to the weaker \emph{time-average} convergence that is typically obtained in regret minimization).
\Cref{e:regr-supp-1l} and \Cref{i:regr-ret-1} give regret guarantees for both the supplier and the retailer.
This shows that not only the utilities of both the supplier and the retailer converge to that of the SE, but that both their cumulative utilities match (up to lower-order terms) those that would be gathered by consistently selecting the SE $(w^\star,q^\star)$ at all time steps with full knowledge of $\cD$.
Crucially, note that while for the supplier it is possible to obtain finite-time regret guarantees, for the retailer these only hold asymptotically.
\begin{proof}
By \Cref{t:unique-stack}, under \Cref{ass:distribution-learning-strong}, the best response function $\BR$ defined in \eqref{e:best-response} maps each wholesale price $w$ into its unique best-response $q^\star_w=\BR(w)$, which the retailer (R) can compute by \Cref{ass:limited-info}. 
Since R is best responding, the utility of the supplier (S), for
any $w > 0$, is $q^\star_w\brb{ w-\E[C] }$.
Note that, under \Cref{ass:distribution-learning-strong}, the unique SE is $(w^\star,q^\star_{w^\star})$, where $w^\star$ is the unique maximizer of $q^\star_w\brb{ w-\E[C] }$, which S cannot compute directly because \Cref{ass:limited-info} is not sufficient for S to determine $q^\star_w$.
However, S can calculate $\E[C]$ with the knowledge of the marginal distribution of $C$.
Hence, S gets a noise-free evaluation $q^\star_{w_t}\brb{ w_t-\E[C] }$ at each round $t$, \emph{after} selecting the wholesale price $w_t$ for the round.

Now, note that S's objective $w\mapsto q^\star_w\brb{ w-\E[C] }$ is \lip{}. Indeed, recalling~\eqref{e:zeroth-order-supplier}, \eqref{e:derivative-of-supllier-objective}, \eqref{e:derivative-of-h}, and \Cref{ass:distribution-learning-strong}, we get that, for any $w\in(0,1)$,
\[
    \labs{ \frac{\partial}{\partial w}q^\star_w\brb{ w-\E[C] } }
\le
    \frac{1}{\E[P] L} \cdot \brb{ 1 - \E[C] } + 1.
\]
Since for any time horizon $T$, S selects the best point in a grid of step-size at most $T^{-1/2}$, we have that 
\begin{equation}
\label{e:conve-ret}
    \lim_{T\to\iop}\labs{ w_T - w^\star} 
\le 
    \lim_{T\to\iop} T^{-1/2} 
= 
    0.
\end{equation}
Then, using again \eqref{e:derivative-of-supllier-objective}, \eqref{e:derivative-of-h} and \Cref{ass:distribution-learning-strong}, we get that R's best response function $\BR$ is $\fracc{1}{(\E[P] L)}$-\lip{}.
Recalling \eqref{e:derivative-of-retailer-objective}, we also have that for any fixed wholesale price $w$, R's instantaneous utility at time $t$, $q\mapsto \E \bsb{ \rho(q;w,P_t,D_t)}$ is $1$-\lip{}.
Hence, for any time horizon $T$, we have that
\begin{equation}
    \label{e:conv-suppl}
    \lim_{T\to\iop}\labs{ q_T - q^\star_{w^\star}} 
\le 
    \lim_{T\to\iop}  T^{-1/2} / (\E[P] L)
= 
    0.
\end{equation}
Putting \eqref{e:conv-suppl} and \eqref{e:conve-ret} together, gives \Cref{i:last-iterate-conv-1}.
\Cref{e:regr-supp-1l} is an immediate consequence of the $\brb{ \fracc{(1-\E[C])}{(\E[P]L)} +1 }$-\lip{}ness of S's objective $w\mapsto q^\star_w\brb{ w-\E[C] }$.
\Cref{i:regr-ret-1} is an immediate consequence of the $(L^{-1}+1)$-\lip{}ness of R's utility $w \mapsto \E \bsb{ \rho(q^\star_w;w,P_t,D_t)}$ as a function of the wholesale price (which follows directly from the chain rule).
\end{proof}
The previous result yields sublinear regret guarantees for the supplier even when S is oblivious to the expected retail price $\E[P]$ and the lower bound $L$ on the conditional density of the demand given the retail price.
Since the \lip{} constant of S's objective is a deterministic function of $\E[C]$, $L$, and $\E[P]$, the reader might wonder if improved regret guarantees could be achieved if these quantities were known to S.
We show  
now 
that this is indeed the case.

We refine \Cref{ass:limited-info} as follows.
\begin{assumption}
\label{ass:limited-info+lip} 
The supplier has access to the $\E[C]$, $L$, and $\E[P]$. Moreover, the retailer has access to the marginal distribution of $(P,D)$.
\end{assumption}
Under Assumptions~\ref{ass:distribution-learning-strong}~and~\ref{ass:limited-info+lip}, and given that the R best-responds to any wholesale price $w$, the supplier can compute their expected cost and solve their zeroth-order \lip{} optimization problem with the knowledge of (an upper bound of) the \lip{} constant of its objective.
This can be done with the \piya{} algorithm (\Cref{a:piya}).
We now assume that this is S's strategy.

\begin{algorithm}
\DontPrintSemicolon
\caption{\piya{}\label{a:piya}}
\textbf{input:} Time horizon $T$, \lip{} constant $M>0$\;

\textbf{initialization:} Let $w_1 \coloneqq 1$\;
\For{$t=1,\dots,T$}
{
    Select the wholesale price $w_t$\;
    Observe the quantity $q_t$\;
    Update the proxy function $\hat \rho _t(\cdot) \coloneqq \min_{s\in[t]} \bcb{ q_t \brb{ (\cdot)-\E[C] } + M \bno{ w_s - (\cdot) } }$\;
    Let $w_{t+1} \in \argmax_{w\in [0,1]} \hat \rho_t(w)$\;
}
\end{algorithm}

The \piya{} algorithm has been known for half a century \cite{piyavskii1972algorithm,shubert1972sequential}, but only recently it was proven that it enjoys outstanding theoretical guarantees for its query complexity, regret, and robustness \cite{bouttier2022regret,gokcesu2021regret}.
In particular, the following theorem follows directly by specializing \cite[Theorem~3.5]{bouttier2022regret} and \cite[Theorem~1]{gokcesu2021regret} to our setting.
\begin{theorem}[\cite{bouttier2022regret,gokcesu2021regret}]
\label{t:piya}
Under Assumptions~\ref{ass:distribution-learning-strong}~and~\ref{ass:limited-info+lip}, for any horizon $T$, if the supplier runs \Cref{a:piya} with inputs $T$ and $M \coloneqq \fracc{(1-\E[C])}{(\E[P]L)} +1 $, and the retailer best responds,
then the function $w\mapsto q^\star_w\brb{ w - \E[C] }$ is $M$-\lip{} and, for all $t \in[T]$,
\begin{align*}
    \max_{w\in[0,1]}\bcb{ q^\star_w\brb{ w - \E[C] } } - q^\star_{w_t}\brb{ w_t - \E[C] } 
& \le 
    9M\frac{\log_2(Mt)}{t}
\\
    \max_{w\in[0,1]}\bcb{ q^\star_w\brb{ w - \E[C] } } - \frac1T \sum_{t=1}^T q^\star_{w_t}\brb{ w_t - \E[C] } 
& \le
    2M \frac{ \ln(4T) }T
\end{align*}
\end{theorem}
\Cref{t:piya} allows us to prove the following result.

\begin{theorem}
\label{t:strong-assumptions+lip}
Under Assumptions~\ref{ass:distribution-learning-strong}~and~\ref{ass:limited-info+lip}, for any horizon $T$, if the supplier runs \Cref{a:piya} with inputs $T$ and $M \coloneqq \fracc{(1-\E[C])}{(\E[P]L)} +1 $, and the retailer best responds,
then:
\begin{equation}
\label{i:regr-supp-2l} 
        \E \bsb{ \sigma(w^\star;q^\star,C) } - \frac1T \sum_{t=1}^T \E \bsb{ \sigma(w_t;q_t,C_t)} 
    \le 
        2M \frac{ \ln (4T)}{T}
\end{equation}
where $(w^\star,q^\star)$ is the unique SE of the stage game.
Moreover,
\begin{enumerate}
    \item \label{i:regr-ret-2} 
    $
        \lim_{T\to \iop} \Brb{ \E \bsb{ \rho(q^\star;w^\star,P,D) } - \frac1T \sum_{t=1}^T \E \bsb{ \rho(q_t;w_t,P_t,D_t) } }
    = 
        0
    $
    \item \label{i:last-iterate-conv-2} 
    $
        \lim_{T\to\iop} \bno{ (w^\star,q^\star) - (w_T,q_T) }_1 
    =
        0
    $
\end{enumerate}
\end{theorem}

\begin{proof}
Proceeding as in the proof of \Cref{t:strong-assumptions} and applying \Cref{t:piya}, we get that the retailer's instantaneous utility at time $t$, $q\mapsto \E \bsb{ \rho(q;w,P_t,D_t)}$ is $1$-\lip{} for any fixed wholesale price $w$, and the supplier's instantaneous utility at time $t$, $w\mapsto q^\star_w\brb{ w - \E[C_t] }$ is $M$-\lip{}, where $q^\star_w$ is defined as in \eqref{e:best-response}, for all $w\in[0,1]$.
As above, under \Cref{ass:distribution-learning-strong}, the unique SE is precisely $(w^\star,q^\star_{w^\star})$, where $w^\star$ is the unique maximizer of $q^\star_w\brb{ w-\E[C] }$.
Applying again \Cref{t:piya}, we obtain immediately the result.
\end{proof}

\paragraph{Convergence to SE with no information.}
\label{s:convergence-no-info}
If the retailer had access to the distribution $\cD$, or at least to the marginal distribution of $(P,D)$, they could best-respond to the retailer's move $w_t$ at each time $t$, as described in \Cref{s:stage-full-section,s:partial-info}.
Since, in this section, none of these is available to R, we assume that the retailer acts according to the next-best available strategy, i.e., best-responding to an empirical distribution that can be maintained by gathering samples.
In the online learning literature, this strategy is known as Follow-the-Leader (FTL) and is detailed in \Cref{a:ftl}. 
\begin{algorithm}
\DontPrintSemicolon
\caption{Follow-the-Leader (FTL)\label{a:ftl}}
\textbf{input:} Time horizon $T \ge 12$\;
\textbf{initialization:} Let $\cQ \coloneqq \bcb{ \fracc{1}{(\lceil T^{1/3} \rceil + 1)}, \dots, \fracc{\lceil T^{1/3} \rceil}{(\lceil T^{1/3} \rceil + 1)} }$\;
Observe the wholesale price $W_1$\;
Draw a quantity $Q_1$ from $\cQ$ uniformly at random\;
Observe the demand $D_1$ and the retail price $P_1$\;
\For{$t=2,3,\dots$}
{
    Observe the wholesale price $W_t$\;
    Select a quantity $Q_t \in \argmax_{q\in\cQ} \brb{\frac{1}{t-1}\sum_{s=1}^{t-1} \min\{q,D_s\}P_s- q W_t}$\;
    Observe the demand $D_t$ and the retail price $P_t$\;
}
\end{algorithm}
For $t=1$, R picks a quantity at random and observes
the demand $D_{1}$. 
During each time step $t\ge2$, define, for all $w,q\in[0,1]$, the auxiliary function
\[
\widehat{\rho}_{t}(w,q):=\frac{1}{t-1}\sum_{s=1}^{t-1}\min\{q,D_{s}\}P_s-qw.
\]
Note that this is not built to maximize the empirical average of the utility gained in last $t-1$ interactions, i.e., it differs from
$
    q \mapsto \frac{1}{t-1}\sum_{s=1}^{t-1} \brb{ \min\{q,\, D_{s}\}P_s-q W_{s} }
$.
Indeed, the retailer is not interested in maximizing their expected utilities at time steps $t$ but rather, their expected utility given $W_t$.
Equivalently stated, the retailer is not maximizing an expected revenue computed with respect to the empirical distribution of $(C_1,P_1,D_1,W_1)$, $\dots$, $(C_{t-1},P_{t-1},D_{t-1},W_{t-1})$ at time $t$, but rather, that of $(C_1,P_1,D_1,W_t)$, $\dots$,  $(C_{t-1},P_{t-1},D_{t-1},W_t)$ given $W_t$.
This way, $\widehat{\rho}_{t}(w,q)$ is an unbiased estimate of $\E\bsb{\rho(q;w,P_t,D_t)}$ for all $w,q \ge 0$ and $t\ge 2$, which in turn implies that
$\E\bsb{\widehat{\rho}_{t}(W_{t},q) \mid W_{t}}=\E\bsb{\us(q;W_t,P_t,D_t) \mid W_t}$ for all $q\ge 0$ and $t\ge 2$.
This corresponds precisely to the instantaneous objective of the retailer.
Therefore, naturally, the choice of a discretized retailer at time $t\ge2$ is 
\begin{align*}
Q_{t} \in \argmax_{q\in\cQ}\brb{\widehat{\rho}_{t}(W_{t},q)}
=\argmax_{q\in\cQ}\lrb{\frac{1}{t-1}\sum_{s=1}^{t-1}\min\{q,\, D_{s}\}P_s-q W_{t}}.
\end{align*}
Similarly to the previous section, we assume here that the supplier adopts an Explore-Then-Commit strategy (\Cref{a:etc-two}), with the \emph{caveat} that, in this section, the expected production cost $\E[C]$ is not available to S but has to be estimated.
\begin{algorithm}
\DontPrintSemicolon
\caption{Explore-Then-Commit (without knowledge of $\E[C]$)\label{a:etc-two}}
\textbf{input:} Time horizon $T \ge 12$\;
\For{$t=1,\dots,\lceil T^{1/3} +1 \rceil$}
{
    \For{$s=1,\dots,\lceil T^{1/3} \rceil$}
    {
    Select the wholesale price $W_{(t-1)\lceil T^{1/3} \rceil + s} \coloneqq \fracc{s}{\big(\lceil T^{1/3} \rceil + 1\big)}$\;
    Observe the quantity $Q_{(t-1)\lceil T^{1/3} \rceil + s}$ and production cost $C_{(t-1)\lceil T^{1/3} \rceil + s}$\;
    }
}
Compute $S^\star \in \argmax_{ s = \lceil T^{1/3} \rceil ^2+1,\ldots,\lceil T^{1/3} +1 \rceil \lceil T^{1/3} \rceil } Q_s \brb{ W_s - \frac{1}{ \lceil T^{1/3} \rceil ^2 } \sum_{j=1}^{ \lceil T^{1/3} \rceil ^2 } C_j }$\;
\For{$t= \lceil T^{1/3} +1 \rceil \lceil T^{1/3} \rceil+1,\dots,T$}
{
    Select the wholesale price $W_t \coloneqq W_{S ^\star}$\;
    Observe the quantity $Q_t$\;
}
\end{algorithm}

\begin{theorem}
\label{th:etc-ftl}
Under \Cref{ass:distribution-learning-strong}, for any horizon $T \ge 12$, if the supplier runs Explore-Then-Commit (\Cref{a:etc-two}) with input $T$ and the retailer runs Follow-the-Leader (\Cref{a:ftl}) with input $T$, then:
\begin{equation}
\label{e:regr-supp-no-assl} 
        \E \bsb{ \sigma(w^\star;q^\star,C) } - \frac1T \sum_{t=1}^T \E \bsb{ \sigma(W _t;Q_t,C_t) } 
    \le 
        \lrb{
        16 
        +
        \frac{1-\E[C]}{\E[P]L}
        +
        7 \sqrt{ \ln T }
    }T^{-1/3},
\end{equation}
where $(w^\star,q^\star)$ is the unique SE of the stage game.
Moreover:
\begin{enumerate}
    \item \label{i:regr-ret-no-ass} 
    $
        \lim_{T\to\iop} \Brb{ \E \bsb{ \rho(q^\star;w^\star,P,D) } - \frac1T \sum_{t=1}^T \E \bsb{ \rho(Q_t;W_t,P_t,D_t \mid W_t) } }
    =
    0
    $ 
    w.p.\ $1$.
    \item \label{i:last-iterate-conv-no-ass}
    $
    	\lim_{T\to\iop}
        \bno{ (w^\star,q^\star) - (W_T,Q_T) }_1 
    =
    0
    $
    with probability $1$.
\end{enumerate}
\end{theorem}
\begin{proof}
Fix any time horizon $T\ge 12$.
Proceeding as in the proof of \Cref{t:strong-assumptions} and applying \Cref{t:piya}, we get that the retailer's instantaneous utility at time $t$, $q\mapsto \E \bsb{ \rho(q;w,P_t,D_t)}$ is $1$-\lip{} for any fixed wholesale price $w$, 
and the supplier's instantaneous utility at time $t$, $w\mapsto q^\star_w\brb{ w - \E[C_t] }$ is $M$-\lip{}, where $q^\star_w$ is defined as in \eqref{e:best-response}, for all $w\in[0,1]$, and $M \coloneqq \fracc{\brb{1-\E[C]}}{\brb{\E[P]L}} +1$.
As above, under \Cref{ass:distribution-learning-strong}, the unique SE is precisely $(w^\star,q^\star_{w^\star})$, where $w^\star$ is the unique maximizer of $w\mapsto q^\star_w\brb{ w-\E[C] }$.
Now, fix an arbitrary $\delta\in\brb{ 0,1 / \lceil T^{1/3} +1 \rceil }$.
Observing that for any $t\ge 2$, given $W_t$, the retailer's quantity $Q_t$ is the $\argmax$ of an empirical average translated by a constant, applying Hoeffding's inequality $\lceil T^{1/3} +1 \rceil$ times and the fact that the retailer's discretization has step-size $1/\lceil T^{1/3} +1 \rceil$, we obtain that
\begin{align}
	\label{e:conc-1}
    \labs{ Q_s - q^\star_{s/(\lceil T^{1/3} \rceil + 1)} }
& \le
    \sqrt{ \frac{\ln \nicefrac2\delta }{ 2 \lceil T^{1/3} \rceil ^2 } }
    +
    \frac{1}{\lceil T^{1/3} + 1 \rceil}
\\
&
    \qquad\text{for all $s = \lceil T^{1/3} \rceil ^2 + 1, \dots , \lceil T^{1/3} + 1 \rceil \lceil T^{1/3} \rceil$} \nonumber
\\
	\label{e:conc-2}
    \labs{ \frac{1}{ \lceil T^{1/3} \rceil ^2 } \sum_{j=1}^{ \lceil T^{1/3} \rceil ^2 } C_j - \E[C] }
& \le
    \sqrt{ \frac{\ln \nicefrac2\delta }{ 2 \lceil T^{1/3} \rceil ^2 } }
\end{align}
hold simultaneously with probability at least $1 - \lceil T^{1/3} +1 \rceil\delta$.
Thus, 
\[
    \labs{ Q_s \lrb{ W_s - \frac{1}{ \lceil T^{1/3} \rceil ^2 } \sum_{j=1}^{ \lceil T^{1/3} \rceil ^2 } C_j } - q^\star_{W_s} \brb{ W_s -\E[C] } } 
\le
    2\sqrt{ \frac{2 \ln \nicefrac2\delta }{ \lceil T^{1/3} \rceil ^2 } }
    +
    \frac{3}{\lceil T^{1/3}+1 \rceil}
\]
hold simultaneously  for all $s = \lceil T^{1/3} \rceil ^2 + 1, \dots , \lceil T^{1/3} + 1 \rceil \lceil T^{1/3} \rceil$, with probability at least $1 - \lceil T^{1/3} +1 \rceil\delta$.
Consequently, letting 
\[
	\tilde w \in \argmax_{w\in \{ 1/\lceil T^{1/3} + 1 \rceil, \dots, \lceil T^{1/3} \rceil/\lceil T^{1/3} + 1 \rceil \} } q^\star_w \brb{ w - \E[C] }
\]
and given that the supplier's discretization has a step size $1/\lceil T^{1/3}+1 \rceil$ and $w\mapsto q^\star_w\brb{ w - \E[C] }$ is $M$-\lip{}, the triangular inequality yields, for any $t\ge \lceil T^{1/3} +1 \rceil \lceil T^{1/3} \rceil+1$,
\begin{align*}
&
    \labs{ Q_t \lrb{ W_t - \frac{1}{ \lceil T^{1/3} \rceil ^2 } \sum_{j=1}^{ \lceil T^{1/3} \rceil ^2 } C_j } - \E\bsb{ \sigma (w^\star, q^\star, C)} }
\le
\\
& \quad
\le
	\labs{ Q_t \lrb{ W_{S^\star} - \frac{1}{ \lceil T^{1/3} \rceil ^2 } \sum_{j=1}^{ \lceil T^{1/3} \rceil ^2 } C_j } - q^\star_{W_{S^\star}} \brb{ W_{S^\star} - \E[C] } } 
\\
& \quad
	+ \Babs{ q^\star_{W_{S^\star}} \brb{ W_{S^\star} - \E[C] } - q^\star_{\tilde w} \brb{ \tilde w - \E[C] } } + \Babs{ q^\star_{\tilde w} \brb{ \tilde w - \E[C] } - \E\bsb{ \sigma (w^\star, q^\star, C) } }
\\
& \quad
\le
    6\sqrt{ \frac{2 \ln \nicefrac2\delta }{ \lceil T^{1/3} \rceil ^2 } }
    +
    \frac{9+M}{\lceil T^{1/3}+1 \rceil}
\end{align*}
with probability at least $1 - \lceil T^{1/3} +1 \rceil\delta$.
Charging regret $1$ to the supplier for the first $\lceil T^{1/3} +1 \rceil \lceil T^{1/3} \rceil$ rounds, then summing the previous bound over all remaining rounds and upper bounding $T-\lceil T^{1/3} +1 \rceil \lceil T^{1/3} \rceil$ with $T$ yields
\begin{align*}
&
    \E \bsb{ \sigma(w^\star;q^\star,C) } - \frac1T \sum_{t=1}^T \sigma(W_t;Q_t,C_t)
\\
&\qquad\le
    \frac{\lceil T^{1/3} +1 \rceil \lceil T^{1/3} \rceil}T
    +
	6\sqrt{ \frac{2 \ln \nicefrac2\delta }{ \lceil T^{1/3} \rceil ^2 } }
    +
    \frac{9+M}{\lceil T^{1/3}+1 \rceil}
\\
&
\qquad\le
    \brb{
        15 
        +
        M
        +
        6 \sqrt{ 2 \ln \nicefrac2\delta }
    }T^{-1/3}
\end{align*}
with probability at least $1 - \lceil T^{1/3} +1 \rceil\delta$.
Thus, \eqref{e:regr-supp-no-assl} follows directly by choosing, e.g., $\delta= 2T^{-2/3}$ and upper bounding $12/\sqrt{3}$  with  $7$.
The proof of \Cref{i:last-iterate-conv-no-ass} is a simple consequence of \eqref{e:conc-1} and the fact that the retailer's best response function $\BR$ is \lip{}.
Finally, \Cref{i:regr-ret-no-ass} follows directly by \eqref{e:conc-1}, \eqref{e:conc-2}, and \lip{}ness of the retailer's utility $w \mapsto \E \bsb{ \rho(q^\star_w;w,P_t,D_t)}$ as a function of the wholesale price (which is implied by the chain rule).
\end{proof}
Again, the previous result shows  \emph{last-iterate} convergence of the supply chain to the unique SE (in contrast to the weaker \emph{time-average} convergence that is typically obtained in regret minimization).
In contrast to our previous results, this theorem holds under much weaker assumptions on the prior knowledge of R and S.
Indeed, we do not assume anything other than the knowledge that the support of $\cD$ is included in $[0,1]^3$.

\section{Vertical Integration}
\label{se:vertical}
Our supply chain model can be cast in different market scenarios. For example, the supply chain can be vertically integrated or not, and the retailer can be a price-taker or a price-maker. In \Cref{s:stage-full-section,s:online-learning} we studied the case of a price-taking retailer without vertical integration. In this section, instead, we study a model with a price-making retailer and a vertically integrated supply chain. 
This means that the supplier sells to the retailer at production cost, and the entire burden of maximizing the utility of the pair is delegated to the retailer, who is also choosing the retail price.
Note that vertical integration goes beyond the case of a company controlling their supply chain.
Indeed, as we mentioned in the introduction, a popular strategy in supply-chain revenue-maximization is to sign revenue-sharing contracts in which the supplier sells to the retailer at production cost, but then the profit of the retailer is split between the two according to previously-agreed percentages \cite{CacLar:MS2005}.

The vertical integration assumption simplifies the interactions between the supplier and retailer, and allows in turn to generalize the learning protocol we discussed in \Cref{s:online-learning} to a more challenging setting.
In particular, we are be able to study an adversarial (rather than a stochastic) environment and a censored feedback (where the retailer does not get to see the demand if it was higher than the quantity they purchased).
Moreover, we consider the more complex case of a price-taking retailer, who is simultaneously optimizing both the quantity and the retail price (in contrast to the previous sections, where the retail price was determined exogenously and the retailer only had to optimize over one variable).

More precisely, each instance of our problem is characterized by two arbitrarily chosen \emph{unknown} sequences.
The first one is a sequence $(c_t)_{t\in\N}$ of real numbers in $[0,1]$, representing the production \emph{costs of the supplier}, and the second one is a sequence $(d_t)_{t\in \N}$ of non-increasing $[0,1]$-valued functions defined on $[0,1]$, representing the \emph{demand} of the market as a function of the retail price.

We study the following online protocol.
At each round $t=1,2,\dots$:
\begin{enumerate}
\item The supplier (S) reveals the wholesale price $w_t \coloneqq c_t$ to the retailer (R).
\item R buys a (possibly random) quantity $Q_t \in [0,1]$ paying $Q_t c_t$ to S.
\item R selects a (possibly random) retail price $P_t \in [0,1]$.
\item The market buys a quantity $\min\{Q_t,d_t(P_t)\}$, paying $\min\{Q_t,d_t(P_t)\}P_t$ to R.
\end{enumerate}
Similarly to \Cref{s:online-learning}, we define the auxiliary function
\[
	\rho(p,q;w,d)
\coloneqq
	\min\{q,d\} p - qw
	\quad
	\forall (p,q,w,d) \in [0,1]^4.
\]
Note that, in this case, we do not require an auxiliary function to compute the utility of the supplier because their individual utility is zero by definition.
Crucially, this implies that the social welfare coincides with the utility of the retailer.
The goal of the supply-chain is then to minimize, for any time horizon $T$, the regret
\[
    R_T
\coloneqq
	\sup_{(p,q) \in [0,1]^2} \sum_{t=1}^T \E \Bsb{ \rho \brb{ p,q;c_t,d_t(p) } }
	-
	\sum_{t=1}^T \E \Bsb{ \rho \brb{ P_t,Q_t;c_t,d_t(P_t) } },
\]
where the expectation is with respect to the algorithm's internal randomization.
In words, this corresponds to a long-term social welfare that is as close as possible to what an omniscient learner could achieve, if they had perfect \emph{a priory} knowledge of the sequences $(c_t)_{t\in\N}, (d_t)_{t\in\N}$ and unlimited computing power.

Consider now a special case of the above protocol, where the demand function $d_t$ at each round $t$ takes the form
\begin{align*}
    d_t(p) = 
    \begin{cases}
    1 & \text{if } p\le v_t, \\
    0 & \text{if } p > v_t.
    \end{cases}
\end{align*}
Here $(v_t)_{t\in \N}$ is an unknown arbitrary sequence of real numbers in $[0,1]$, representing the market valuation for a unit quantity of the good. In this case, the choice $q=1$ is trivially optimal, and so the retailer must only choose the retail price $p_t$. When $q_t=1$, the terms $q_t c_t$ cancels out in the regret, and so we can write the social welfare as
\begin{align*}
    \rho(p,1;c_t,d_t)
    = p \, \min\{1, d_t(p)\}
    = 
    \begin{cases}
    p & \text{if } p \le v_t, \\
    0 & \text{if } p > v_t.
    \end{cases}
\end{align*}
This shows that a special case of our online protocol is the adversarial posted-price problem studied in~\cite{kleinberg2003value}, where they prove a lower bound on the regret of order $\Omega(T^{2/3})$.

\paragraph{Worst-case analysis of the retailer's regret.}
In this section, we show that for adversarial demand and production cost, the supply-chain regret is $R_T = \widetilde{\cO}\big(T^{2/3}\big)$, thus matching (up to logarithmic factors) the lower bound for the posted-price problem.

It is easy to see that the function $\rho$ satisfies the following \lip{} conditions for all $p,q,c,d \in [0,1]$,
\begin{align*}
	\rho(p+\delta,q;c,d) &\le \rho(p,q;c,d) + \delta \qquad \forall \delta \in [0, 1-p],
\\
	\rho(p,q+\delta;c,d) &\le \rho(p,q;c,d) + \delta \qquad \forall \delta \in [0, 1-q].
\end{align*}
We introduce the \expthree{} algorithm, a variant of the Exp3 algorithm for multi-armed bandits \cite{auer2002nonstochastic} adapted to exploit the richer feedback available in the supply chain setting---see also \cite{cesa2017algorithmic} for a similar application of Exp3 to second-price auctions. 
The algorithm uses a discretization of the action space $[0,1]^2$ in $K(K+1)$ actions $(p'_i,q'_j)$, where $i\in[K]$, $j\in[K+1]$, and $K = \lceil 1/\gamma \rceil$, for some $\gamma > 0$. 
We set $p'_k,q'_k = (k-1)\gamma$ for $k\in[K]$ and $q'_{K+1} = 1$. Because of the \lip{} conditions, for any $c,d\in[0,1]$,
\begin{equation}
\label{eq:lip}
	\max_{p,q \in [0,1]} \sum_{t=1}^T \rho(p,q;c,d)
\le
	\max_{1 \le i,j \le K} \sum_{t=1}^T \rho(p'_i,q'_j;c,d) + 2\gamma\,T.
\end{equation}
In order to simplify the presentation of the analysis, $[-1,1]$-valued revenues are turned into $[0,1]$-valued losses. 
We set, for all $t\in \N$, $i\in [K]$, and $j\in[K+1]$
\[
	\loss_t(i,j) 
\coloneqq 
    \frac{1 - \rho_t \brb{ p'_i,q'_j,c_t,d_t(p'_i) } }{2}.
\]
\begin{algorithm}
\DontPrintSemicolon
\caption{\expthree{}\label{a:expthree}}
\textbf{input:} Time horizon $T$, exploration parameter $\gamma > 0$, learning rate $\eta > 0$\;

\textbf{initialization:} Set $K \coloneqq \lce{ 1/\gamma }$ and the uniform distribution $\pi_1$ over the decision space $V \coloneqq [K] \times [K+1]$\;
\For{$t=1,\dots,T$}
{
    Observe the supplier's cost $c_1$\;
    Draw $(I_t,J_t) \sim \mu_t$, where the distribution $\mu_t$ is defined by
    \[
        \mu_t(i,j) 
    \coloneqq
        (1-\gamma) \pi_t(i,j) + \frac{\gamma}{K} \I \{j=K+1\} \qquad \forall (i,j) \in V
    \]
        
    Select the retail price $P_t \coloneqq (I_t-1)\gamma$ and the quantity $Q_t \coloneqq (J_t-1)\gamma$\;
    Observe the censored demand $\min\bcb{ Q_t, d_t(P_t) }$\;
    For each $(i,j) \in V$, compute the estimated loss
	\[
		\hloss_t(i,j) = \frac{\loss_t(i,j)}{\sum_{k=j}^{K+1} \mu_t(i,k)} \I \{ i=I_t,\, j \leq J_t \}
	\]
	
	For each $(i,j) \in V$, compute the new probability assignment
	\[
		\pi_{t+1}(i,j) = \frac{\exp\big(-\eta \sum_{s=1}^t \hloss_s(i,j)\big)}{\sum_{(m,n) \in V} \exp\left(-\eta\sum_{s=1}^t \hloss_s(m,n)\right)}
	\]
}
\end{algorithm}
Note that \expthree{} (Algorithm~\ref{a:expthree}) chooses $(p_t,q_t)$ without using the knowledge of $c_t$; 
as a consequence, it could be that $p_t < c_t$.
In words, it is possible that the supply chain sells at a loss in some time steps.
This is a common occurrence that sharply distinguishes repeated from one-shot settings, especially in non-stochastic environments in which some unpredictability of the learner is necessary to combat the adversarial behavior of Nature.

We are now ready to prove the main theorem of this section.
\begin{theorem} \label{th:buyer}
For any time horizon $T$, the \expthree{} algorithm run with inputs $T$, $\gamma > 0$, and $\eta > 0$ satisfies 
\begin{equation}
\label{e:upper-bound-vertical-integration}
    R_T
\le
    \eta KT\ln\frac{eK}{\gamma} + \frac{4\ln(K+1)}{\eta} + 4\gamma T.
\end{equation}
In particular, if $\gamma \coloneqq T^{-1/3}$ and $\eta \coloneqq T^{-2/3}$ , then
\[
	R_T 
\le 
    3 ( 4 + 3 \ln T ) T^{2/3}
\]
(more accurate choices of $\gamma$ and $\eta$ could lead to better leading constants).
\end{theorem}
The proof follows the same lines as the regret analysis of Exp3. The key change is a tighter control of the variance term allowed by the feedback structure.
For all technical details, see \Cref{s:vert-int-appe}.

% bibliography
 \bibliographystyle{plainnat}
 \bibliography{regret}

\clearpage

\appendix

\section{Missing proofs and technical details}
\label{s:proof:unique-stack}

\subsection{Proof of Lemma~\ref{l:proprerties-of-g-2}}
\label{s:proof-lemma-g}

By \Cref{ass:distribution-stage} combined with the \leib{} integral rule, the function $x\mapsto \E \bsb{ P \bar F(x \mid C,P) }$ is differentiable, with derivative 
\begin{equation}
    \label{e:derivative-of-h}
    x
\mapsto
    \frac{\partial}{\partial x} \E \bsb{ P \bar F(x \mid C,P) }
=
    - \E \bsb{ P f(x \mid C,P) }
\end{equation}
To verify that the derivative is strictly negative, simply note that \Cref{ass:distribution-stage} states that $f(d \mid c,p)>0$ for all $(c,p,d)\in [0,\iop)^3$ and $\E[P] > \E[C] \ge 0$; hence the random variable $P$ is strictly positive on a set with strictly positive measure.
This implies that $h(x)$ in \eqref{eq:hx} is strictly decreasing, and therefore invertible.
By the inverse function theorem, we conclude that the inverse $h^{-1}$ is also differentiable with strictly negative derivative, hence strictly decreasing.

\subsection{A sufficient condition for Theorem~\ref{t:unique-stack}}

\begin{proposition}
The second condition of \Cref{t:unique-stack} is satisfied when $(C,P)$ is deterministic (i.e., $(C,P)=(c,p)$ for $0<c<p$) and $D$ has a Weibull distribution with nondecreasing failure rate.  That is, the cumulative distribution function of $D$ is
\begin{align*}
    F(x) & = 1 - e^{-(\lambda/x)^k}
\end{align*}
for $x\ge 0$, where $\lambda>0$ and $k\ge 1$.
\label{p:unique}
\end{proposition}
\begin{proof}
It suffices to show that the following $L(w)$ is concave for $w\in(c,p)$:
\begin{align*}
    L(w) & = g(w) \, (w-c)
\end{align*}
Since the second derivative is
\begin{align*}
    L''(w)
    & = g''(w) \, (w - c) + 2 \, g'(w)
\end{align*} 
and
\begin{align*}
    h(x)
    & = p \, \bar F(x)
    = p \, e^{-(x/\lambda)^k}\\
    g(w)
    & = h^{-1}(w)
    = \lambda \, (\ln (p/w))^{1/k}\\
    g'(w)
    & = -\frac{\lambda}{k\,w} \, (\ln(p/w))^{1/k-1}\\
    g''(w)
    & = \frac{\lambda}{k\,w^2} \, (\ln(p/w))^{1/k-1}
    - \frac{\lambda\,(k-1)}{k^2\,w^2} (\ln(p/w))^{1/k-2},
\end{align*}
we have
\begin{align*}
    L''(w)
    & = - \frac{\lambda}{k^2\,w^2}  \, (\ln(p/w))^{1/k-2} \, \left(
        k \, (w + c) \, \ln(p/w)
        + (k-1)\,(w - c)
    \right)\\
    & \le 0
\end{align*}
where the last inequality holds when $\lambda>0$, $k\ge 1$, $w\in(c,p)$.
\end{proof}

\subsection{Price of anarchy of the \stack{} Equilibrium}
\label{s:inefficiency}

In \Cref{s:online-learning}, we presented several algorithmic ideas that allow supplier and retailer to gain an amount of revenue close to that of the SE.
This is possible despite the two players not cooperating with each other but trying to instead maximize their own individual utilities.
Given that, even in this circumstance, the utility at the equilibrium is achievable, it is now natural to investigate if the equilibrium is efficient.
In this section, we will show that, in general, this is not the case.

Consider the supply-chain stage game introduced in \Cref{s:stage-full-section}. 
Fix any $0<c<p<1$ and let $\cD \coloneqq \delta_c \otimes \delta_p \otimes \cU$ be the product distribution of a Dirac at $c$, a Dirac at $p$, and a uniform distribution $\cU$ on $[0,1]$.
Note that $\cD$ satisfies our base \cref{ass:distribution-stage}, with $f(d) \equiv f(d \mid c,p) \equiv 1$ , for all $d \in [0,1]$.
Let, as usual, $F(d) = \int_0^d f(x) \dif x = d$, for all $d\in [0,1]$.

In words, we are considering the one-shot supply-chain game, where the supplier's cost $c$ and the retail price $p$ are deterministic and the demand $D$ is uniform on $[0,1]$, and we are assuming that $c,p,\cU$ and $\cD$ are common knowledge to both players.

In this instance, the function $h$ ---recall \eqref{eq:hx}--- is defined, for any $x\in[0,1]$, by
$
    h(x)
=
    p - p x
$
and its inverse $g \equiv h^{-1}$ is given, for any $w \in [0,p]$, by
$
    g(w)
=
    \fracc{(p-w)}{p}
$.

We show now that the condition in \Cref{i:condition-unique-stack} of \Cref{t:unique-stack} holds.
Indeed, for any $w \in (c,p)$, 
\[
    \underbrace{-\frac{ -1/p }{ \fracc{(p-w)}{p} } }_{-\fracc{ g'(w) } { g(w) }}
=
    \frac{1}{w-c}
\iff
    \frac{w-c}{p}
=
    \frac{p-w}{p}
\iff
    w
=
    \frac{c+p}2
    \in (c,p)
\]
Hence, the set $\argmax_{w \in \{(\fracc{c+p)}{2}\}}g(w)(w-c)$ is trivially a singleton, and, applying \Cref{t:unique-stack}, the supply-chain stage game admits the unique SE
\[
    (w^\star,q^\star)
=
    \lrb{ \frac{c+p}2, \, \frac{p-c}{2p} }
\]
The expected social welfare is, for any strategy profile $(w,q)\in[0,1]^2$,
\[
    \us(w,q) + \ur(w,q) 
= 
    p \E[\min(q,D)] -cq
=
    p q (1-\fracc q2) - cq
\]
which, under vertical integration, is maximized at $(\tilde w, \tilde q)$, where $\tilde w$ is any wholesale price in $[0,1]$ and
\[
    \tilde{q} 
= 
    \frac{p-c}{p}
\]
Note that, in equilibrium, the retailer orders a quantity that is smaller than the quantity that would be ordered under vertical integration.
Moreover, under vertical integration, the optimal social welfare would be
\[
    \us(\tilde w,\tilde q) + \ur(\tilde w,\tilde q)
= 
    \frac{(p-c)^{2}}{2p}
\]
while the equilibrium social welfare is
\[
    \us(w^\star,q^\star) + \ur(w^\star,q^\star)
=
    \frac{3}{4} \frac{(p-c)^{2}}{2p}.
\]
Thus, in this game the price of anarchy is $4/3$.

\subsection{Proof of Theorem~\ref{th:buyer}}
\label{s:vert-int-appe}

Before stating the main result of this section, we prove an auxiliary lemma.

\begin{lemma}
\label{lem:exp3}
Let $T\in \N$, $V$ be a finite set of cardinality $K$, and fix any sequence $\loss_1,\ldots,\loss_T$ of nonnegative functions $\loss_t : V\to\R$. 
Fix $\eta > 0$ and let $w_1,\ldots,w_T$ be functions $w_t \colon V \to \R$ such that, for any $i \in V$,
\[
	w_t(i) 
= 
    \begin{cases}
        1 & \text{ if $t=1$}
    \\
		\exp \brb{ -\eta\sum_{s=1}^{t-1}\loss_s(i) } & \text{ otherwise}
    \end{cases}
\]
Then, for all $k \in V$,
\begin{align*}
    \sum_{t=1}^{T} \sum_{i \in V} p_{t}(i) \loss_t(i) - \sum_{t=1}^{T} \loss_t(k)  
\le
    \frac{\ln K}{\eta}
    + \frac{\eta}{2} \sum_{t=1}^{T} \sum_{i \in V} p_{t}(i) \loss_t(i)^{2}
\end{align*}
where $p_t(i) = w_t(i)/W_t$ and $W_t = \sum_{j \in V} w_t(j)$ for all $i \in V$ and $t \in [T]$.
\end{lemma}
\begin{proof}
For all $t\in [T]$, we have
\begin{align*}
    \frac{W_{t+1}}{W_t}
&=
    \sum_{i \in V} \frac{w_{t+1}(i)}{W_t}
=
    \sum_{i \in V} \frac{w_{t}(i)}{W_t}\,\exp\bigl(-\eta\,\loss_t(i)\bigr)
=
    \sum_{i \in V} p_{t}(i)\,\exp\bigl(-\eta\,\loss_t(i)\bigr)
\\ &\le
    \sum_{i \in V} p_{t}(i)\left(1 - \eta\,\loss_t(i) + \frac{\bigl(\eta\,\loss_t(i)\bigr)^2}{2}\right)
    \tag{using $e^{-x} \leq 1-x+x^2/2$ for all $x \ge 0$}
\\ &\le
    1 - \eta\sum_{i \in V} p_{t}(i)\loss_t(i)
    + \frac{\eta^2}{2}\sum_{i \in V} p_{t}(i)\loss_t(i)^2
\end{align*}
Taking logs, upper bounding $x\mapsto\log (1+x)$ with $x\mapsto x$, and summing over $t\in[T]$ yields
\begin{equation}
    \label{e:aaa}
    \ln\frac{W_{T+1}}{W_1}
\le
    - \eta\sum_{t=1}^T \sum_{i \in V} p_{t}(i)\loss_t(i)
    + \frac{\eta^2}{2} \sum_{t=1}^T\sum_{i \in V} p_{t}(i)\loss_t(i)^2
\end{equation}
Moreover, we also have, for any $k\in V$,
\begin{equation}
    \label{e:bbb}
    \ln\frac{W_{T+1}}{W_1}
\ge
    \ln\frac{w_{T+1}(k)}{W_1} = -\eta\sum_{t=1}^T \loss_t(k) - \ln K
\end{equation}
Putting \eqref{e:aaa} and \eqref{e:bbb} together, dividing both sides by $\eta > 0$, and rearranging gives the desired result.
\end{proof}

We can now prove \Cref{th:buyer}.

\begin{proof}[Proof of \Cref{th:buyer}]
We first control the regret associated with actions drawn from $\pi_t$ (the regret associated with $\mu_t$ will be studied as a direct consequence). Fix a time horizon $T$, any $(m,n) \in V$, and an arbitrary sequence of realizations $(I_1,J_1),\ldots,(I_T,J_T)$. This determines the nonnegative estimated losses $\hloss_t(i,j)$ and we can apply Lemma~\ref{lem:exp3} to the probabilities $\pi_t$ to get
\begin{equation}
\label{eq:boundestimatedloss}
	\sum_{t=1}^T \sum_{(i,j) \in V} \pi_{t}(i,j) \hloss_t(i,j) - \sum_{t=1}^T \hloss_t(m,n)
\le
	\frac{\eta}{2} \sum_{t=1}^T \sum_{(i,j) \in V} \pi_t(i,j) \hloss_t(i,j)^2  + \frac{2\ln(K+1)}{\eta}
\end{equation}
Writing $\E_{t-1}[\cdot]$ for the expectation conditioned on $(I_1,J_1),\dots,(I_{t-1},J_{t-1})$, we note that, for any $(i,j) \in V$, $\E_{t-1}\lsb{ \hloss_t(i,j) } = \ell_t(i,j)$ (i.e., $\hloss_t(i,j)$ is an unbiased estimate of $\ell_t(i,j)$) and
\[
	\E_{t-1}\Big[ \pi_t(i,j) \hloss_t(i,j)^2\Big] = \frac{\pi_t(i,j) \ell_t(i,j)^2}{\sum_{k=j}^{K+1} \mu_t(i,k)} 
\leq 
    \frac{\mu_t(i,j)}{(1-\gamma)\sum_{k=j}^{K+1} \mu_t(i,k)}
\]
where we used the definition of $\mu_t$ and the fact that $\ell_t(i,j)^2 \leq 1$ by construction. Therefore, taking expectations on both sides of~\eqref{eq:boundestimatedloss} implies
\begin{multline*}
	\E\left[\sum_{t=1}^T \sum_{(i,j) \in V} \pi_{t}(i,j) \loss_t(i,j)\right] - \sum_{t=1}^T \loss_t(m,n)
\\
\le
	\frac{\eta}{2(1-\gamma)} \sum_{t=1}^T \sum_{i=1}^K \E\left[\sum_{j=1}^{K+1} \frac{\mu_t(i,j)}{\sum_{k=j}^{K+1} \mu_t(i,k)}\right]  + \frac{2\ln(K+1)}{\eta}    
\end{multline*}
For any $t\in [T], i\in[K], j\in[K+1]$, set $s_t(j) \coloneqq \sum_{k=j}^{K+1} \mu_t(i,k)$.
Using $s_t(j)$, we can upper bound the sum inside the expectation with an integral.
For any fixed $i\in[K]$, we have
\begin{align*}
&	\sum_{j=1}^{K+1} \frac{\mu_t(i,j)}{\sum_{k=j}^{K+1} \mu_t(i,k)}
=
	1 + \sum_{j=1}^K \frac{s_t(j)-s_t(j+1)}{s_t(j)}
=
	1 + \sum_{j=1}^K \int_{s_t(j+1)}^{s_t(j)} \frac{\dd x}{s_t(j)}
\\&\le
	1 + \sum_{j=1}^K \int_{s_t(j+1)}^{s_t(j)} \frac{\dd x}{x}
=
	1 + \int_{\mu_t(i,K+1)}^{s_t(1)} \frac{\dd x}{x}
\le
	1 - \ln \mu_t(i,K+1)
\le
	1 + \ln \frac{K}{\gamma}
\end{align*}
where we used $s_t(1) \le 1$ and $\mu_t(i,K+1) \ge \gamma/K$. 
Therefore, substituting into the previous bound, we get
\begin{equation}
	\E\left[\sum_{t=1}^T \sum_{(i,j) \in V} \pi_{t}(i,j) \loss_t(i,j)\right] - \sum_{t=1}^T \loss_t(m,n)
\le
	\frac{\eta KT}{2 (1-\gamma)}\ln\frac{eK}{\gamma} + \frac{2\ln(K+1)}{\eta}
\label{eq:Exp3Floor-regretbeforemixing}
\end{equation}
We now control the regret of $(I_t,J_t)$ drawn from $\mu_t$. We have
\begin{align*}
&	\E\left[\sum_{t=1}^T \ell_t(I_t,J_t)\right] - \sum_{t=1}^T \ell_t(m,n)
\\&=
	\E\sum_{t=1}^T \left[ \sum_{i=1}^K \left( (1-\gamma) \sum_{j=1}^{K+1} \pi_t(i,j) \loss_t(i,j) + \frac{\gamma}{K} \ell_t(i,K+1) \right)\right] - \sum_{t=1}^T \ell_t(m,n)
\\&\le
	(1-\gamma) \E\left[\sum_{t=1}^T \sum_{(i,j) \in V} \pi_t(i,j) \loss_t(i,j) \right] + \gamma T - \sum_{t=1}^T \ell_t(m,n)
\\&\le
	\frac{\eta KT}{2}\ln\frac{eK}{\gamma} + \frac{2\ln(K+1)}{\eta} + \gamma T
\end{align*}
where the last inequality is by~\eqref{eq:Exp3Floor-regretbeforemixing}.

Translating back from losses to revenues, and using~\eqref{eq:lip}, gives
\[
	R_T
\le
	\eta KT\ln\frac{eK}{\gamma} + \frac{4\ln(K+1)}{\eta} + 4\gamma T
\]
Recalling that $K = \lceil 1/\gamma \rceil$, and choosing $\gamma = T^{-1/3}$ and $\eta = T^{-2/3}$ concludes the proof.
\end{proof}

\end{document}